\documentclass[a4paper,11pt]{article}
\usepackage{rotating}
\usepackage{fullpage}
\usepackage{algorithm}
\usepackage{algcompatible}
\usepackage{graphics}
\usepackage{setspace}
\usepackage{hyperref}
\usepackage{relsize}
\usepackage{hyperref}
\hypersetup{colorlinks=true,linkcolor=blue,citecolor=blue}
\usepackage{color}
\usepackage{amsthm}
\usepackage{graphicx}
\usepackage{amsmath}
\usepackage{amssymb}
\usepackage{diagbox}
\usepackage{caption}
\usepackage[classfont=bold,
langfont=roman,
funcfont=italic]{complexity}

\theoremstyle{plain}
\newtheorem{theorem}{\textbf{Theorem}}
\newtheorem{proposition}[theorem]{\textbf{Proposition}}
\newtheorem{claim}[theorem]{\textbf{Claim}}
\newtheorem{property}[theorem]{\textbf{Property}}
\newtheorem{conjecture}[theorem]{\textbf{Conjecture}}
\newtheorem*{claim*}{\textbf{Claim}}
\newtheorem*{theorem*}{\textbf{Theorem}}
\newtheorem{definition}[theorem]{\textbf{Definition}}
\newtheorem{lemma}[theorem]{\textbf{Lemma}}
\newtheorem{corollary}[theorem]{\textbf{Corollary}}

\newcommand\ones{\texttt{\upshape s}_{1}}
\newcommand\zeros{\texttt{\upshape s}_{0}}

\newcommand\zerobs{\texttt{\upshape bs}_{0}}
\newcommand\bs{\texttt{\upshape bs}}
\newcommand\s{\texttt{\upshape s}}
\newcommand\oneC{\texttt{\upshape C}_{1}}
\newcommand\dv{d_{\vee}}
\newcommand\da{d_{\wedge}}
\newcommand{\abs}[1]{\left\lvert #1 \right\rvert}
\title{\textbf{On the Sensitivity Conjecture for \\ Disjunctive Normal Forms}}
\date{}
\author{%
Karthik C.\ S.\thanks{This work was partially supported by Irit Dinur's ERC-StG grant number 239985.  Some parts of this work were done while interning at Microsoft Research, India.} \\
Weizmann Institute of Science \\
\texttt{karthik.srikanta@weizmann.ac.il}.
\and S\'ebastien Tavenass\thanks{This work was supported by ANR project CompA (project number: ANR-13-BS02-0001-01).}\\ Microsoft Research India\\ \texttt{t-sebat@microsoft.com}. 
}

\begin{document}
\maketitle 
\renewcommand{\thefootnote}{\fnsymbol{footnote}}

\begin{abstract}
\begin{sloppypar}The sensitivity conjecture of Nisan and Szegedy [CC '94] asks whether for any Boolean function $f$, the maximum sensitivity $\s(f)$, is polynomially related to its block sensitivity $\bs(f)$, and hence to other major complexity measures. Despite major advances in the analysis of Boolean functions over the last decade, the problem remains widely open.\\ 
\end{sloppypar} 

In this paper, we consider a restriction on the class of Boolean functions through a model of computation (DNF), and refer to the functions adhering to this restriction as admitting the Normalized Block property. We prove that for any function $f$ admitting the Normalized Block property, $\bs(f) \leq 4\s(f)^2$. 
We note that (almost) all the functions mentioned in literature that achieve a quadratic separation between sensitivity and block sensitivity admit the Normalized Block property.   \\
 
Recently, Gopalan et al.\ [ITCS '16] showed that every Boolean function $f$ is uniquely specified by its values on a Hamming ball of radius at most $2\s(f)$. We extend this result and also construct examples of Boolean functions which provide the matching lower bounds.
\end{abstract}

\clearpage
\section{Introduction} 

Sensitivity and block sensitivity are complexity measures that are commonly used
for Boolean functions. Both these measures were originally introduced
for studying the time complexity of CRAW-PRAM’s \cite{CD82,CDR86,N91}. Block
sensitivity is polynomially related to a number of other complexity measures,
such as the decision-tree complexity, the certificate complexity, the
polynomial degree, and the quantum query complexity \cite{BD02}.
A longstanding open problem is the relation between sensitivity and block sensitivity.
From the definitions of sensitivity and block sensitivity, it immediately follows
that $\s(f) \le \bs(f)$, where $\s(f)$ and $\bs(f)$ denote the sensitivity and
the block sensitivity of a Boolean function $f$. Nisan and Szegedy \cite{NS92} conjectured
that sensitivity is also polynomially related to
block sensitivity:
\begin{conjecture}[Sensitivity Conjecture \cite{NS92}]
There exist constants $\delta,c>0$ such that for every Boolean function $f$ we have that $\bs(f)\le c\cdot\left(\s(f)\right)^{\delta}$.
\end{conjecture}
This conjecture is still widely open and the best known upper bound on block sensitivity is exponential in terms of sensitivity \cite{ABGMSZ14}. On the other hand, the best known separation (through an example of a Boolean function) between sensitivity and block sensitivity is quadratic \cite{AS11}; more background and discussion about the sensitivity conjecture can be found in the survey of Hatami et al.\  \cite{HKP11}. 

Over the last decade, in the majority of the works concerning the sensitivity conjecture, the focus has been on addressing the conjecture for restricted classes of Boolean functions, where the restriction is imposed by some notion of symmetry \cite{S05, S07, D11}. The reason behind pursuing this direction is that nonconstant Boolean functions with a high degree of symmetry must have high complexity according to various measures. Accordingly, all the results in this direction \cite{S05, S07, D11} show that the sensitivity of the corresponding functions is large (in terms of the number of variables), and deduce that the sensitivity is close to block sensitivity. While we feel that proving the sensitivity conjecture for a restricted class of Boolean functions is a step in the right direction, we would like to argue that these specific restrictions are limited in their potential to explicitly promote the understanding of the \emph{relationship} between sensitivity and block sensitivity.

In this paper, we prove the sensitivity conjecture for a restricted class of Boolean functions, where the restriction is imposed on a DNF representation of the function. This is one of the first time[s] since Nisan \cite{N91} that the sensitivity conjecture is proved for a restriction based on a model of computation (recently, Lin and Zhang \cite{LZ16} proved the sensitivity conjecture for functions admitting circuits with a small number of negation gates, and in a simultaneous work~\cite{BLTV16}, the authors prove the sensitivity conjecture in the case of regular read-$k$ formulas of constant depth with $k$ constant). 
Informally, the restriction we impose on the DNF can be described as follows. We assume that the maximal block sensitivity is reached on the all zeroes input and that the function outputs a zero on this input,  and notice that for each clause in the DNF, the set of  positive literals in the clause corresponds to a sensitive block. Based on the fact that the block sensitivity counts the number of disjoint sensitive blocks, we consider the natural restriction where the set of positive literals of each of the clauses are also disjoint.  We say that any function adhering to this restriction admits the \emph{normalized block property}, and we show that for any Boolean function $f$ admitting the normalized block property, $\bs(f) <4\s(f)^2$.  

As the other side of the same coin, this result provides a barrier to building Boolean functions with super-quadratic separation between sensitivity and block sensitivity. Currently, the best known separation is given by an example of Ambainis and Sun \cite{AS11} who built a function $f$ with $\bs(f) = \frac{2}{3}\s(f)^2-\frac{1}{3}\s(f)$. Ambainis and Sun additionally showed that their example gives the best possible separation (up to an additive factor) between sensitivity and block sensitivity for all functions that are an OR of functions whose zero-sensitivity equals 1. We build a framework (of restrictions) over DNFs and identify where the result of Ambainis and Sun lies within this framework, and our result that the sensitivity conjecture is true for Boolean functions admitting normalized block property is shown to be an extension of the result of Ambainis and Sun. Additionally, Kenyon-Kutin \cite{KK04}, showed that if the block sensitivity is attained on some input which has blocks of size at most two then, $\bs\le e\cdot \s^2$. More generally,

\begin{theorem}[Kenyon and Kutin\cite{KK04}]\label{KenyonKutin}
For every Boolean function $f$ on $n$ variables, and every $\ell \in \{2,\ldots ,\s(f)\}$, we have:
$$\bs_\ell(f)\le\frac{e}{(\ell-1)!} (\s(f))^\ell,$$
where $\bs_\ell(f)$ is the block sensitivity of $f$ when each block is restricted to be of size at most $\ell$.
\end{theorem}

	Therefore, to construct examples of Boolean function with super-quadratic separation between sensitivity and block sensitivity we now have two barriers. Moreover, we extend the notion of block property to $t$-block property, and prove a lower bound on the sensitivity of Boolean functions admitting the $t$-block property in terms of $t$, and the width and size of the DNF. 	 

Recently, Gopalan et al.\ \cite{GNSTW16} investigated the computational complexity of low sensitivity functions and provided interesting upper bounds on their circuit complexity. This was indicated to be a promising alternative approach to the sensitivity conjecture as opposed to getting improved
bounds on specific low level measures like block sensitivity or decision tree depth \cite{KK04,ABGMSZ14,AS11}. In particular, they showed that every Boolean function $f$ is uniquely specified by its values on a Hamming ball of radius at most $2\s(f)$, and showed various applications of this result. We extend this result by showing that if two Boolean functions $f$ and $g$ coincide on a ball of radius $\s(f)+\s(g)$ then, $f=g$. Furthermore, for every $p,q>1$, we construct examples of Boolean functions  $f$ and $g$ such that $\s(f)=p$, $\s(g)=q$, and $f$ and $g$ coincide on a ball of radius $\s(f)+\s(g)-1$ but $f\neq g$, showing that the above result is tight.  

Finally, we propose a computational problem motivated by the sensitivity conjecture, and the existing work and results therein. Assuming the sensitivity conjecture to be true, we note that this problem is in \TFNP, and  wonder if resolving the sensitivity conjecture would yield an efficient algorithm to this computational problem.   

This paper is organized as follows. In Section~\ref{Preliminaries}, we provide the basic definitions of complexity measures, structures, and objects that will be used in the rest of the paper. In Section~\ref{sec_BP}, we define a few restrictions (such as the block property) on DNFs representing Boolean functions and prove the sensitivity conjecture for the class of functions admitting (some of) these structural restrictions. In Section~\ref{Low}, we investigate a structural result of low sensitivity functions. In Section~\ref{Problem}, we propose a new computational problem motivated by the sensitivity conjecture. Finally, in Section~\ref{Conclusion}, we conclude with a promising open question on proving the sensitivity conjecture for functions admitting the $t$-block property.

\section{Preliminaries}\label{Preliminaries}

We use the notation $[n]=\{1,\ldots ,n\}$. Let $f:\{0,1\}^n\to\{0,1\}$, be a Boolean function. Let $x\in\{0,1\}^n$. For $i\in [n]$, we denote by $x^i$ the input in $\{0,1\}^n$ which is obtained by flipping the $i^{\text{th}}$ bit of $x$. Also for any $B\subseteq [n]$, we denote by $x^B$ the input in $\{0,1\}^n$ which is obtained by flipping the bits of $x$ in all coordinates in $B$. We will now define two complexity measures on Boolean functions which are of great interest. 
\begin{definition}
The sensitivity of
a Boolean function $f$ at input $x \in \{0, 1\}^n$, written $\s(f, x)$, is the number of coordinates  $i\in[n]$  such that $f (x) \neq f (x^i)$. The sensitivity of $f$, written $\s(f )$, is defined as $\s(f ) =\underset{ x\in\{0,1\}^n}{\max} \s(f, x)$. We define $\ones(f) = \underset{ f(x)=1}{\max}\s(f,x) $ and $\zeros(f) = \underset{ f(x)=0}{\max}\s(f,x) $.
\end{definition}

\begin{definition}
\sloppypar{
The block sensitivity of a Boolean function $f$ at input $x\in\{0, 1\}^n$, for $k$ disjoint subsets $B_1,\ldots ,B_k$ of $[n]$ (called blocks), written $\bs(f, x,B_1,\ldots ,B_k)$, is the number of blocks  $i\in[k]$  such that $f (x) \neq f (x^{B_i})$. The block sensitivity of a Boolean function $f$ at input $x \in \{0, 1\}^n$, written as $\bs(f, x)$, is the maximum of $\bs(f, x,B_1,\ldots ,B_k)$ over all $k$ disjoint subsets $B_1,\ldots ,B_k$ of $[n]$ for all $k\in[n]$. The block sensitivity of $f$, written $\bs(f )$, is defined as $\bs(f ) =\underset{ x\in\{0,1\}^n}{\max} \bs(f, x)$.  We define $\bs_1(f) = \underset{ f(x)=1}{\max}\bs(f,x) $ and $\bs_0(f) = \underset{ f(x)=0}{\max}\bs(f,x) $.}
\end{definition}

We will now introduce a model of representation of Boolean functions. 

\begin{definition}
A DNF (disjunctive normal form) formula $\Phi$ over Boolean variables $x_1, \dots, x_n$ is defined to be a logical $\mathrm{OR}$ of terms, each of which is a logical $\mathrm{AND}$ of literals. A literal is either a variable $x_i$ or its logical negation $\overline{x}_i$. We insist that we can assume that no term contains both a variable and its negation (otherwise we can remove this term). We often identify a DNF formula $\Phi_f$ with the Boolean function $f : \{0,1\}^n \to \{0,1\}$ it computes.
\end{definition}

We note here that for every Boolean function $f$, there exists at least one (it is not unique) DNF formula $\Phi_f$ that computes it. 

\section{Block Property}
\label{sec_BP}

In the following, we will often use the notation $\vee$ (respectively $\wedge$) for denoting the Boolean operation OR (respectively AND).
Let $f$ be a Boolean function and $\Phi_f$ be one of its DNF formulas.
Let $X=\{x_1,\ldots ,x_n\}$ be the set of variables. Let $\dv$ be the fan-in of the $\vee$-gate which is usually called the size of the DNF. We label the $\dv$ $\wedge$-gates as: $\wedge_1,\ldots, \wedge_{\dv}$. Let $d_{\wedge_i}$ be the fan-in of $\wedge_i$. Let $\da=\underset{i}{\max}\ d_{\wedge_i}$ be the width of the DNF. For every $i\in[\dv]$, let $A_i$ be the set of variables amongst the literals connected to $\wedge_i$ appearing \textbf{without} a negation and let $\overline{A}_i$ be the set of variables amongst the literals connected to $\wedge_i$ appearing \textbf{with} a negation.
An assignment of the variables is a function $\sigma:X\to\{0,1\}$. For every $\wedge_i$, we define $S_i$ as follows:
$$S_i=\{\sigma\mid\wedge_i(\sigma)=1\},$$
where $\wedge_i(\sigma)$ is the evaluation of $\wedge_i$ when the assignment to the variables is given by $\sigma$. 

By negating some variables and/or negating the output of the function, we can always assume that the maximum block sensitivity is the maximum $0$-block sensitivity (i.e., $\bs_0$) and is reached on the all zeros input. Moreover, given a DNF representation of our function, we can assume that this representation is minimal (i.e., any subformula of the given formula computes a distinct function). 
\begin{definition}\label{cf}
A Boolean function $f$ represented by a DNF formula $\Phi_f$ is said to be represented in \textbf{compact form} if the following holds:
\begin{enumerate}
\item[a)] $f(0^n)=0$,
\item[b)] The maximum $0$-block sensitivity is attained on the all zeroes input, i.e., $\zerobs(f)=\bs(f,0)$,
\item[c)] and $\forall i\in[\dv]$, we have that $S_i\setminus \underset{\begin{subarray}{c}
j=1, \\
  j\neq i
  \end{subarray}}{\overset{\dv}{\bigcup}}S_{j}\neq \emptyset$.
\end{enumerate}

Moreover the representation is called  \textbf{normalized} if the maximal block sensitivity is also attained on the all zeroes input, i.e., $\bs(f)=\bs(f,0)$.
\end{definition}

The condition c) means that for each $i$ there exist a $\sigma$ such that  $\sigma$ makes only $\wedge_i$ true.

\begin{lemma}
  For every $f:\{0,1\}^n\to\{0,1\}$, there exists $f^\prime:\{0,1\}^n\to\{0,1\}$ such that $\s(f^\prime)=\s(f)$, $\bs({f^\prime})=\bs(f)=\bs(f^\prime,0^n)$, and $f^\prime$ admits a normalized compact form representation.
\end{lemma}

\begin{proof}
We claim that for any Boolean function $f$, there exists another Boolean function $f^\prime$ such that $\s(f)=\s(f^\prime)$, $\bs(f)=\bs(f^\prime)$, $f^\prime(0^n)=0$ and such that $f^\prime$ attains its maximal block sensitivity at the all zeroes input. This is because, if $f$ attains its maximum block sensitivity at $a\in\{0,1\}^n$ then, we define $f^\prime(x)=f(a)\oplus f(x\oplus a)$\footnote{The operator $\oplus$ denotes the usual $\mathrm{XOR}$ function.}, and the claim follows.  

Let us fix a DNF formula for $f^\prime$. If there is $i\in [\dv]$ such that $ S_i \subseteq \bigcup_{j\neq i} S_j$, then we do not change the function by removing $\wedge_i$. Thus any such AND gates can be assumed to have been removed.
\end{proof}

In fact, we can remark that only the condition $f(0^n)=0$ from Definition~\ref{cf} may need a larger DNF (since, it could need to compute the negation of the original function), other constraints can be achieved without increasing the size of the formula. 

We will now describe a structural result about Boolean functions that admit compact form representation. For every $i\in[\dv]$, we define $\Gamma_i$ as follows: $$\Gamma_i=\left\{j\Big| \left\lvert A_i\cap \overline{A}_j\right\rvert+\left\lvert A_j\cap \overline{A}_i\right\rvert=1\right\}.$$

Informally, $\Gamma_i$ is the set of AND gates which contradict on $\wedge_i$ on exactly one variable. Let $\Gamma=\underset{i}{\max}\ |\Gamma_i|$. We bound $\ones$ using $\Gamma$ as follows:
\begin{lemma}\label{gamma}
Any Boolean function $f$ represented in the compact form admits the following bound on $\ones$: $$\da-\Gamma\le \ones\le \da.$$
\end{lemma}
\begin{proof}
First, we prove that $\ones\le \da$. Let $a\in\{0,1\}^n$ be the input for which the maximum $\ones$ is attained. By definition of $\ones$, we have that $f(a)=1$. Let $\wedge_i$ be an AND gate such that $\wedge_i(a)=1$. Suppose $\ones>\da$ then there exists $x_j\in X\setminus(A_i\cup \overline{A}_i)$ such that $f(a^j)=0$. But, as $\wedge_i$ does not depend on $x_j$, $\wedge_i(a^j)=1$ and so $f(a^j)$ still equals $1$, which is a contradiction.  

We will now prove that $\ones\ge\da-\Gamma$. Let $i_0=\underset{i}{\text{argmax }}d_{\wedge_{i}}$. Let $b\in S_{i_0}\setminus \underset{\begin{subarray}{c}
j=1, \\
  j\neq i_0
  \end{subarray}}{\overset{\dv}{\bigcup}}S_{j}$ (from Definition~\ref{cf}c such a selection is possible). We have that $\wedge_{i_0}(b)=1$ and for all $j\in [\dv]\setminus\{i_0\}$, $\wedge_j(b)=0$. It is sufficient to lower bound the cardinality of $C\subseteq[n]$ such that for all $i\in C$, we have that $f(b^i)=0$. Fix some $x_k\in (A_{i_0} \cup \overline{A}_{i_0})$. We observe that $f(b^k)=1$ implies that there is an AND gate $\wedge_j$ such that $\left(A_{i_0}\cap \overline{A}_j\right)\cup\left(\overline{A}_{i_0}\cap {A}_j\right)=\{x_k\}$. There are exactly $|\Gamma_{i_0}|$ such $k$'s. The lower bound follows.
\end{proof}

Nisan~\cite{N91} showed that for all monotone functions the block sensitivity and sensitivity are equal. This was the first time that the sensitivity conjecture was proven for a class of functions captured by a restriction on the model of computation for Boolean functions. In our setting, Nisan's result would be written as follows:
\begin{theorem}[Nisan \cite{N91}]
Let $f$ be a Boolean function and $\Phi_f$ be a compact form representation of $f$. In $\Phi_f$ if for every $i\in \dv$, we had that $\overline{A}_i=\emptyset$ then, $\bs(f)=\s(f)$.
\end{theorem}
In this paper, we look at Boolean functions through weaker restrictions on their DNF representation. In this regard, we will now see three kinds of structural impositions on Boolean functions in compact form representation. Later, we will prove the sensitivity conjecture for the class of functions admitting (some of) these structural impositions. 

\begin{property}[Block property]
A Boolean function is said to admit the block property if under a compact form representation  $\forall i,j\in [\dv]$ such that $i \neq j$, we have that $A_i\cap A_j=\emptyset$. Moreover, if there exists such a compact form representation which is also normalized, we will say that the function admits the normalized block property.
\end{property}

\begin{property}[Mixing property]
A Boolean function is said to admit the $\ell$-mixing property if under a compact form representation $\forall i,j\in [\dv]$ with $i\neq j$, such that if $\left(A_i\cup \overline{A}_i\right)\cap\left(A_j\cup \overline{A}_j\right)\neq\emptyset$ we have, $\left\lvert\left(A_i\cap \overline{A}_j\right)\cup\left(A_j\cap \overline{A}_i\right)\right\rvert\ge \ell$.
\end{property}

\begin{property}[Transitive property]
A Boolean function is said to admit the transitive property if under a compact form representation $\forall i,j,k\in [\dv]$, we have that if $(A_i\cup\overline{A}_i)\cap (A_j\cup\overline{A}_j)\neq\emptyset$ and if $(A_j\cup\overline{A}_j)\cap (A_k\cup\overline{A}_k)\neq\emptyset$ then, $(A_i\cup\overline{A}_i)\cap (A_k\cup\overline{A}_k)\neq\emptyset$.
\end{property}

First, we see that if a Boolean function admits the Mixing property then we can improve the bound obtained in Lemma~\ref{gamma}.

\begin{lemma}\label{lem_Gamma0}
Let $\ell>1$. Any Boolean function admitting the $\ell$-mixing property has $\Gamma=0$.
\end{lemma}
\begin{proof}
\begin{sloppypar}
Fix $i\in[\dv]$. Since the function admits $\ell$-mixing property, we know that for every $j\in[\dv]$, either $\left(A_i\cup \overline{A}_i\right)\cap\left(A_j\cup \overline{A}_j\right)=\emptyset$ in which case we have that $j\notin \Gamma_{i}$, or $\left\lvert\left(A_i\cap \overline{A}_j\right)\cup\left(A_j\cap \overline{A}_i\right)\right\rvert\ge \ell$ in which case we  again conclude that $j\notin\Gamma_{i}$ because of the following:\end{sloppypar}
$$
1<\ell\le \left\lvert\left(A_i\cap \overline{A}_j\right)\cup\left(A_j\cap \overline{A}_i\right)\right\rvert
=\left\lvert A_i\cap \overline{A}_j\right\rvert+\left\lvert A_j\cap \overline{A}_i\right\rvert,
$$
where the last equality holds because in the definition of DNF formula we insisted that no term contains both a variable and its negation.
Therefore, we have that $\Gamma_{i}=\emptyset$.
\end{proof}

Consequently, we have that $\ones=\da$, for all Boolean functions admitting the $\ell$-mixing property with $\ell>1$. 

Ambainis and Sun had previously shown in Theorem~2 of~\cite{AS11} that their construction gave the (almost) best possible separation between block sensitivity and sensitivity for a family of Boolean functions.
Let us consider the Boolean functions $f$ which can be written as a variables-disjoint union:
\begin{align*}
f= \bigvee_{i=1}^n g(x_{i,1},\ldots,x_{i,m}). \tag{1} \label{eq:boo}
\end{align*}
Then (see for example Lemma 1 in~\cite{AS11} or Proposition 31 in~\cite{GSS13}) $\ones(f)=\ones(g)$, $\zeros(f)=n\zeros(g)$, and $\zerobs(f)=n\zerobs(g)$. So if we can find a lower bound for the sensitivity of $g$ with respect to $\zerobs(g)$, we get the best gap for $f$ by choosing $n=\ones(g) / \zeros(g)$.  
\begin{theorem}[Ambainis and Sun~\cite{AS11}]\label{AS}
  If $g$ is a Boolean function such that $\zeros(g)=1$ and $\bs(g)=\zerobs(g)$, then 
$$
    \ones(g) \geq 3\frac{\bs(g) -1}{2}.
  $$
\end{theorem}

In fact, we can notice that these functions 
belong to our framework (this claim is implicit in their proof of Theorem~\ref{AS}, but we give a proof in Appendix~\ref{frame}):
\begin{claim}\label{ASinFramework}
 Let $g$ be as in Theorem~\ref{AS}. Let $f$ be the OR of several copies of $g$, where each copy takes its input from a different set of variables, as in Eq.~\eqref{eq:boo}. Then, there exists $f^\prime$ with same block sensitivity and at most same $1$-sensitivity which admits the normalized block property, the transitive property, and the $3$-mixing property.
\end{claim}

Ambainis and Sun~\cite{AS11} present an explicit Boolean function $f$ such that $\bs=\frac{2}{3}\s^2-\frac{1}{3}\s$.  The function is a variables-disjoint union
 $$
    f= \bigvee_{i=1}^{3n+2} g(x_{i,1},\ldots,x_{i,4n+2}).
  $$ The function $g$ outputs one if the $4n+2$ corresponding variables satisfy the pattern $P_{\textrm{AmbainisSun}}$ or if it is the case after an even-length cyclic rotation of the variables.
  The pattern starts with \(2n\) \(0\)s which are followed by a block of two ones and it finishes by \(n\) copies of the block \(\underline{0\_}\) (the underscore means the variable can be $0$ or $1$):
  \begin{center}
    \begin{tabular}{|c|c|cc|c|c|c|c|c|c|c|cc|c|c|}
      \hline
      0 & 0 & \ldots & \ldots & 0 & 1 & 1 & 0 & \_ & 0 & \_ & \ldots & \ldots & 0 & \_
      \\     \hline
    \end{tabular} 
  \end{center}
  As we only admit the even-length rotations, we can easily see that the normalized block property is ensured. The patterns in $g$ pairwise intersect, so we also get the transitive property. Finally, if we consider two rotations $R_1$ and $R_2$ of the pattern, we can assume that the $\underline{11}$-block in $R_1$ intersects a $\underline{0\_}$-block in $R_2$ (otherwise, we switch $R_1$ and $R_2$ and get it). Then the $\underline{11}$-block in $R_2$ will intersect a $\underline{00}$-block in $R_1$. The two rotations of the pattern disagree on at least three variables (and in fact exactly three). Hence the $3$-mixing property is also verified.

We show in Appendix~\ref{sec_gap} that other functions in literature achieving a quadratic gap (eg: Rubinstein \cite{R95}, Virza \cite{V11}, Chakraborthy \cite{S05}) fall in our framework. 

We ended up proving a result which supersedes the one mentioned in Theorem~\ref{AS} both in the lower bound and for a more general family. The above lower bound is \textbf{exactly} matched by the Boolean function constructed by Ambainis and Sun \cite{AS11}. This implies that there cannot exist a Boolean function admitting the normalized block property, the transitive property and the 2-mixing property which has a better separation between block sensitivity and sensitivity than the function constructed by Ambainis and Sun \cite{AS11}.

\begin{theorem}\label{thm_strongAS}
Any Boolean function admitting the normalized block property, the transitive property, and the $2$-mixing property and which depends on at least two variables has:  $$\bs\le \frac{2}{3}\s^2-\frac{1}{3}\s.$$
\end{theorem} 

The proof of the above theorem is in Appendix~\ref{2mixing}. In a previous version of this paper, we did not assume that the number of dependent variables is at least two. However, as Kri\v{s}j\={a}nis Prusis and Andris Ambainis pointed out to us, there was a small error in the proof and indeed the univariate function $f(x)=x$ does not satisfy this inequality ($\s=\bs=1$). Moreover, they noticed that, as the $2$-mixing property implies $\ones = \da=\oneC$ (cf. Lemma~\ref{lem_Gamma0} and the following remark), their result~\cite{AP14} directly implies that any Boolean function admitting the $2$-mixing property satisfies $\bs \leq \frac{2}{3}\s^2 +\frac{1}{3}\s$.  

Our main result is to get rid of the dependence on the transitive property and the mixing property. Imposing only the normalized block property on DNFs is a weak restriction as there is no constraint on $\overline{A}_i$. Further, given the DNF in compact form representation admitting the normalized block property is a natural way to represent the function through its (maximal) block sensitivity complexity. We show the following theorem concerning Boolean functions admitting block property:

\begin{theorem}\label{block}
Any Boolean function admitting the block property has $\zerobs\le 4 \s^2$. In particular, if the representation is normalized, $\bs \leq 4\s^2$.
\end{theorem}

The importance of the result is that the block property seems to be a quite natural restriction for studying the relations between the sensitivity and the block sensitivity. In fact, by assuming that the block sensitivity is maximized, by the blocks $B_i$, on the all zeros inputs with $f(0^n)=0$ (which is always possible), the block property intuitively asserts the output is one if from the all zeros input, we can get an input in $f^{-1}(1)$ only by flipping at least one of the blocks $B_i$. If it is not the case, it would mean there are other non-disjoint blocks which are present just for diminishing the sensitivity.  

Before presenting the proof, we prove three lemmas, after which the above result follows immediately. 

\begin{lemma}\label{bs=dv}
Any Boolean function $f$ admitting the block property has $\zerobs=\dv$.
\end{lemma}
\begin{proof}
From Definition~\ref{cf}a, we have that $f(0^n)=0$ and thus we have that every $A_i$ is non-empty. Now, it is easy to see that $\zerobs\ge \bs(f,0^n) \ge \dv$ -- choose each $A_i$ as a block. Any two blocks are disjoint because of the block property and by flipping any of the blocks, one of the AND gates will evaluate to 1. 

From Definition~\ref{cf}b, we know that the maximum $0$-block sensitivity is attained on $0^n$. Let the sensitive blocks for which it attains maximum \(0\)-block sensitivity be $B_1,\ldots , B_k$. Thus when some $B_i$ is flipped to all $1$s, at least one of the AND gates evaluates to 1. Since the blocks are disjoint, we can associate a distinct AND gate to each sensitive block. Therefore the number of sensitive blocks is at most the number of AND gates, i.e., $\zerobs=k\le \dv$.
\end{proof}

\begin{lemma}\label{zerosbound}
Any Boolean function admitting the block property has: $$\s\ge \left\lceil\frac{\dv}{2\da-1}\right\rceil.$$
\end{lemma}
The previous lemma is easily seen as optimal by a multiplicative factor two by considering the OR function. 
\begin{proof}
Let $E$ be a subset of AND gates such that for any two $\wedge_i,\wedge_j\in E$, we have $A_i\cap \overline{A}_j=\emptyset$ and $A_j\cap \overline{A}_i=\emptyset$. Let $P=\underset{\wedge_i\in E}{\bigcup}P_i$, where $P_i$ is an arbitrarily chosen subset of $A_i$ of size $|A_i|-1$ (note that $|A_i|\ge 1$ as otherwise we would have $f(0^n)=1$, contradicting Definition~\ref{cf}a).  Consider $a\in\{0,1\}^n$, where $a_i=1$ if and only if $x_i\in P$. We observe that for all $\wedge_i\in E$, $\wedge_i(a)=0$. Also, for all  $\wedge_i\notin E$, we have that $A_i\cap P=\emptyset$ from the block property, and therefore $\wedge_i(a)=0$. In short, $f(a)=0$. Now for any $\wedge_i\in E$, let $x_{q(i)}\in A_i\setminus P_i$. Since $\wedge_i(a^{q(i)})=1$, we have that $\zeros(f,a)\ge |E|$.  

Now, we will prove that  there is a set $E$ such that $|E|\ge \left\lceil\frac{\dv}{2\da-1}\right\rceil$. Let $G$ be a directed graph on $\dv$ vertices where the $i^{\text{th}}$ vertex corresponds to $\wedge_i$. We have a directed edge from vertex $i$ to vertex $j$ if $\overline{A}_{i}\cap A_j\neq \emptyset$. Let $U(G)$ be $G$ with orientation on the edges removed. Consider the following procedure for constructing $E$: 
\begin{enumerate}
\item[(1)] Include to $E$, the AND gate corresponding to the vertex with the smallest degree in $U(G)$.
\item[(2)] Remove the vertex picked in (1) and all its in-neighbors and out-neighbors from $G$.
\item[(3)] Repeat (1) if $G$ is not empty. 
\end{enumerate}

From block property, we have that the out-degree of  vertex $i$ in $G$ is at most $\abs{\overline{A}_i}$. Thus the total number of edges in $G$ is at most $\underset{i\in[\dv]}{\sum} |\overline{A}_i| \le \dv(\da -1)$. This implies  that the sum of the degree of all vertices in $U(G)$ is at most $2\dv(\da -1)$.  Therefore, there exists a vertex  in $U(G)$ of degree at most $2\da -2$. By including the corresponding AND gate into $E$, the number of vertices in $G$ reduces by at most $2\da -1$. In order for $G$ to be empty, there should be at least $\left\lceil\frac{\dv}{2\da-1}\right\rceil$ iterations of the above procedure, and since cardinality of $E$ grows by 1 after each iteration, we have that $|E|\ge \left\lceil\frac{\dv}{2\da-1}\right\rceil$.  

Therefore, we have $\s\ge \zeros(f,a)\ge |E|\ge \left\lceil\frac{\dv}{2\da-1}\right\rceil$. 
\end{proof}

\begin{lemma}\label{onesbound}
Any Boolean function admitting the block property has:$$\s\ge\left\lceil\frac{\da}{2}\right\rceil.$$
\end{lemma}

\begin{proof}
If $\ones\ge \left\lceil\frac{1+\da}{2}\right\rceil$, we are done. Therefore, we will assume $\ones<\left\lceil\frac{1+\da}{2}\right\rceil$. Let $i^\star=\underset{i}{\text{argmax}}\ d_{\wedge_i}$. Consider $a\in\{0,1\}^n$ with  $a_j=1$ if and only if $x_j\in A_{i^\star}$. We note that $\wedge_{i^\star}(a)=1$ and $|a|$ (Hamming weight of $a$) is nonzero since $A_{i^\star}$ is nonempty from Definition~\ref{cf}a.  

Let $x_j\in A_{i^\star}$. We claim that $f(a^j)=0$. The proof is by contradiction. Suppose, $f(a^j)=1$. It is clear that $\wedge_{i^\star}(a^j)=0$ as  $x_j\in A_{i^\star}$. Thus, there must exist some $k\neq i^\star$, such that $\wedge_{k}(a^j)=1$. From block property, we know that $A_{i^\star}\cap A_k=\emptyset$, but all variables assigned to 1 in $a^j$ are in $A_{i^\star}$. This implies $A_k=\emptyset$. Therefore $\wedge_k(0^n)=1$, contradicting Definition~\ref{cf}a.   

Now we would like to claim that for any $x_j\in A_{i^\star}$, we have $\zeros(f,a^j)\ge 1+\left\lfloor\frac{\da}{2}\right\rfloor$. We first note that $\ones(f,a)<\left\lceil\frac{1+\da}{2}\right\rceil$ and since for any $x_j\in A_{i^\star}$, we have $f(a^j)=0$, we have that $|A_{i^\star}|<\left\lceil\frac{1+\da}{2}\right\rceil$. Let $D=\{x_p\mid x_p\in \overline{A}_{i^\star}, f(a^p)=1\}$. 
Since,  $\ones(f,a)<\left\lceil\frac{1+\da}{2}\right\rceil$, this implies $|\overline{A}_{i^\star}|-|D|+|A_{i^\star}|<\left\lceil\frac{1+\da}{2}\right\rceil$ or equivalently, $|D|> \left\lceil\frac{\da}{2}\right\rceil -1$. Fix $x_p\in D$ and $x_j\in A_{i^\star}$. Since $f(a^p)=1$, we know there exists some $k\neq i^\star$, such that $\wedge_{k}(a^p)=1$. By block property, we know that $x_j\notin A_k$, and this implies $f(a^{\{j,p\}})=1$. Thus, we have that for any fixed $x_j\in A_{i^\star}$, $\zeros(f,a^j)\ge |D|+1$ as for every $x_p\in D$, we have $f(a^{\{j,p\}})=1$ and also $f(a)=1$. Therefore, for every $x_j\in A_{i^\star}$ we have $\zeros(f,a^j)\ge |D|+1> 1+\left\lceil\frac{\da}{2}\right\rceil-1 =\left\lceil\frac{\da}{2}\right\rceil$. 

Therefore, we have that either $\ones$ or $\zeros$ is at least $\left\lceil\frac{\da}{2}\right\rceil$. 
\end{proof}

\begin{proof}[Proof of Theorem~\ref{block}]
From Lemma \ref{zerosbound} and Lemma \ref{onesbound}, we have that for any Boolean function admitting the block property $\s^2> \frac{\dv}{4}$. Combining this with Lemma~\ref{bs=dv}, we have that $\s^2> \frac{\zerobs}{4}$.
\end{proof}

We can notice that Lemma~\ref{onesbound} is optimal, i.e., we give an example of a Boolean function admitting the block property with $\zeros=\ones=\lceil \da/2 \rceil$. The set of variables is $X=\{x_1, \ldots, x_{2n+1}\}$. We describe the example by its $\wedge$-gates $\wedge_1,\ldots, \wedge_{n+1}$: for all $i\in[n]$, $A_{i}=\left\{x_{2i}\right\}$, $\overline{A}_{i}=\emptyset$, $A_{n+1}=\left\{x_{2i-1}\mid i\in [n+1]\right\}$ and $\overline{A}_{i}=\left\{x_{2i}\mid i\in [n]\right\}$.  

Finally, we conclude with an absolute lower bound on the sensitivity of functions admitting block property.

\begin{corollary}\label{senslower}
  Let $f$ be a Boolean function which depends on $n$ variables. If $f$ admits the block property, then
  $$
    \s(f) \geq \frac{n^{1/3}}{2}.
$$
\end{corollary}

\begin{proof}
  The number of variables which appear in the DNF is at most $\dv \da$, and so $\dv \da \geq n$. By Lemma~\ref{zerosbound} and Lemma~\ref{onesbound},
\begin{equation*}
    \s^3 \geq \left(\frac{\dv}{2\da}\right) \left(\frac{\da}{2}\right)^2 \geq \frac{\dv \da}{8} \geq \frac{n}{8}.\qedhere
\end{equation*}
\end{proof}

\subsection{$t$-Block Property}
In this subsection, we extend the notion of block property to $t$-block property as follows.

\begin{property}[$t$-Block property]
A Boolean function is said to admit the $t$-block property if under a compact form representation  $\forall x\in X$, we have: $$\left|\{A_i|x\in A_i\}\right|\le t.$$
\end{property}

We have that $1$-block property is exactly the same as block property discussed in the previous subsection. Let us notice that the notion of $t$-block property is far more general than the one of read-$t$ DNF presented in~\cite{BLTV16} since, here only the number of times where the variables appear positively is bounded. 

First, we show an upper bound on Boolean functions admitting the $t$-block property in terms of the size of the DNF.

\begin{lemma}\label{bs=tdv}
Any Boolean function $f$ admitting the $t$-block property has $\zerobs\le\dv$.
\end{lemma}
\begin{proof}
From Definition~\ref{cf}b, we know that the maximum $0$-block sensitivity is attained on $0^n$. Let the sensitive blocks for which it attains maximum block sensitivity be $B_1,\ldots , B_k$. We know that $f(0^n)=0$ (from Definition~\ref{cf}a) and thus when some $B_i$ is flipped to all $1$s, the value of the function changes to 1. In other words, at least one of the AND gates evaluates to 1. Since the blocks are disjoint, we can associate a distinct AND gate to each sensitive block. Therefore the number of sensitive blocks is at most the number of AND gates, i.e., $\zerobs=k\le \dv$.
\end{proof}

Next, we prove a lower bound on Boolean functions admitting the $t$-block property in terms of t, the width of the DNF, and the size of the DNF.

\begin{lemma}\label{tzerosbound}
Any Boolean function admitting the $t$-block property has: $$\s\ge \left\lceil\frac{\dv}{3t\da-2t-\da+1}\right\rceil.$$
\end{lemma}
\begin{proof}
Let $E$ be a subset of AND gates such that for any two $\wedge_i,\wedge_j\in E$, we have $ (A_i \cup \overline{A_i}) \cap A_j=\emptyset$. Let $A=\underset{\wedge_i\in E}{\bigcup}A_i$ (note that any $|A_i|\ge 1$ as otherwise we would have $f(0^n)=1$, contradicting Definition~\ref{cf}a). Consider $\mathcal{A}$ the set of $0$-vectors with support in $A$. More formally,
\[ \mathcal{A}= \{ a \in \{0,1\}^n \mid f(a)=0 \textrm{ and }\forall i, a_i=1 \implies x_i\in A \}. \]
First notice that $0^n \in \mathcal{A}$, so this set is not empty. Let $\bar{a}$ be an element of $\mathcal{A}$ with maximal Hamming weight. For any $\wedge_i \in E$, the gate $\wedge_i$ does not depend on the variables in $(A \setminus A_i)$ by definition of $E$ and $A$. So, $f(\bar{a})=0$ implies that there exists a variable $x_{l_i} \in A_i$ such that $\bar{a}_{l_i} = 0$. Then, by maximality of Hamming weight of $\bar{a}$, $f(\bar{a}^{l_i}) =1$ and so $\bar{a}$ is $0$-sensitive on $l_i$. Finally, for all $i,j \in E$ the indices $l_i$ and $l_j$ are distinct since $A_i\cap A_j =\emptyset$. Consequently, we have that $\zeros(f,\bar{a})\ge |E|$.  

Now, we will prove that  there is a set $E$ such that $|E|\ge \left\lceil\frac{\dv}{3t\da-2t-\da+1}\right\rceil$. Let $G$ be a directed graph on $\dv$ vertices where the $i^{\text{th}}$ vertex corresponds to $\wedge_i$. We have a directed edge from vertex $i$ to vertex $j$ if $\overline{A}_{i}\cap A_j\neq \emptyset$. Let $U(G)$ be $G$ with orientation on the edges removed. Consider the following procedure for constructing $E$: 
\begin{enumerate}
\item[(1)] Add to $E$, the AND gate corresponding to the vertex with the smallest degree in $U(G)$.
\item[(2)] Remove the vertex picked in (1) and all its in-neighbors and out-neighbors from $G$.
\item[(3)] Remove any vertex from $G$ associated with a gate $\wedge_j$ with $A_i\cap A_j \neq\emptyset$.
\item[(4)] Repeat from (1) if $G$ is not empty. 
\end{enumerate}

From $t$-block property, we have that the out-degree of  vertex $i$ in $G$ is at most $t\abs{\overline{A}_i}$. Thus the total number of edges in $G$ is at most $\underset{i\in[\dv]}{\sum} t|\overline{A}_i| \le t\dv(\da -1)$. This implies  that the sum of the degree of all vertices in $U(G)$ is at most $2t\dv(\da -1)$.  Therefore, there exists a vertex  in $U(G)$ of degree at most $2t\da -2t$. By including the corresponding AND gate into $E$, the number of vertices in $G$ reduces by at most $2t\da -2t+1$ at step (2). Moreover, at step (3), by the $t$-block property, there are at most $(t-1)\lvert A_i\rvert \leq t\da - \da$ gates $\wedge_j$ such that $A_i\cap A_j \neq\emptyset$ and $j\neq i$. Consequently at most $3t\da -2t-\da+1$ gates are removed at each step. In order for $G$ to be empty, there should be at least $\left\lceil\frac{\dv}{3t\da-2t-\da+1}\right\rceil$ iterations of the above procedure, and since cardinality of $E$ grows by $1$ after each iteration, we have that $|E|\ge \left\lceil\frac{\dv}{3t\da-2t-\da+1}\right\rceil$.  

Therefore, we have $\s\ge \zeros(f,a)\ge |E|\ge \left\lceil\frac{\dv}{3t\da-2t-\da+1}\right\rceil$. 
\end{proof}

As a corollary, we obtain the following.

\begin{corollary}
Let $f$ be a Boolean function admitting the $t$-block property, with $t\le \frac{\dv}{\da^{1+\varepsilon}}$, for some $\varepsilon>0$. Then, we have the following:
$$\zerobs(f)\le t\left(3  \s(f)\right)^{1+\frac{1}{\varepsilon}} .$$
\end{corollary}
\begin{proof}
Since $t\le \frac{\dv}{\da^{1+\varepsilon}}$, we have that $\da\le \left(\frac{\dv}{t}\right)^{{1/(1+\varepsilon)}}$. Substituting in Theorem~\ref{tzerosbound}, we have that $\s(f)\ge \frac{\dv}{3t\left(\frac{\dv}{t}\right)^{{1/(1+\varepsilon)}}}$. After rearranging and simplifying, we get that $t^\varepsilon\left(3\s(f)\right)^{1+\varepsilon}\ge (\dv)^{\varepsilon}$. We substitute Lemma~\ref{bs=tdv}, and simplify to obtain:
\begin{equation*}
t\left(3\s(f)\right)^{1+\frac{1}{\varepsilon}} \ge \zerobs(f). \qedhere
\end{equation*}
\end{proof}

\section{Low Sensitivity Boolean functions}\label{Low}

Gopalan et al.\ \cite{GNSTW16} show that functions with low sensitivity have concise descriptions, so consequently the number of such functions is small. Indeed, they show that knowing the values on a Hamming ball of radius $2\s + 1$ suffices. More precisely,

\begin{theorem}[Gopalan et al.\ \cite{GNSTW16}]\label{funique}
Let $f$ be a Boolean function of sensitivity $\s$. Then, it is uniquely specified by its values on any ball of radius $2\s$.
\end{theorem}

We extend their observation to a more general one:

\begin{theorem}\label{fgunique}
Let $f$ and $g$ be two Boolean functions. If $f$ and $g$ coincide on a ball of radius $\s(f)+\s(g)$ then, $f=g$. 
\end{theorem}

Before we prove Theorem~\ref{fgunique}, we note the following handy lemma:

\begin{lemma}\label{fXORg}
Let $f$ and $g$ be two Boolean functions. We have $\s(f\oplus g)\le \s(f) +\s(g)$ and $\bs(f\oplus g)\le \bs(f) +\bs(g)$ \footnote{$\forall x\in\{0,1\}^n, (f\oplus g)(x)=f(x)\oplus g(x)$.}.
\end{lemma}
\begin{proof}
For any $x\in\{0,1\}^n$ and $i\in [n]$, if $(f\oplus g)(x)\neq(f\oplus g)(x^i)$ then we have that either $f(x)\neq f(x^i)$ or $g(x)\neq g(x^i)$. This implies, for every $x\in \{0,1\}^n$, $\s(f\oplus g,x)\le \s(f,x) +\s(g,x)$. Similarly, for any $x\in\{0,1\}^n$ and $B\subseteq [n]$, if $(f\oplus g)(x)\neq(f\oplus g)(x^B)$ then we have that either $f(x)\neq f(x^B)$ or $g(x)\neq g(x^B)$. This implies, for every $x\in \{0,1\}^n$, $\bs(f\oplus g,x)\le \bs(f,x) +\bs(g,x)$. 
\end{proof}

\begin{proof}[Proof of Theorem~\ref{fgunique}]
The proof is by contradiction. Suppose there exists $a\in \{0,1\}^n$ such that for every $r\in\{0,1\}^n$ of hamming weight at most $\s(f)+\s(g)$, we have that $f(a\oplus r)=g(a\oplus r)$. This implies that for every $r\in\{0,1\}^n$ with $||r||\le \s(f)+\s(g)$, we have $(f\oplus g)(a\oplus r)=0$. Consider $x\in \{0,1\}^n$ of the smallest hamming distance from $a$ such that  $(f\oplus g)(x)=1$. If such a $x$ does not exist then it implies that $f\oplus g$ is the constant zero function. In that case we have that $f=g$, a contradiction. Therefore, let us suppose that $x$ exists as described above. Let $d$ be the hamming distance between $x$ and $a$. We know that $d>\s(f)+\s(g)$. Additionally, we know that there are exactly $d$ neighbors of $x$ at hamming distance $d-1$ from $a$. Since, $x$ was the input with the smallest distance from $a$ such that $(f\oplus g)(x)=1$, we know that the $d$ neighbors of $x$ at hamming distance $d-1$ from $a$ all evaluate to $0$ on $(f\oplus g)$. This means that $\s(f\oplus g,x)\ge d>\s(f)+\s(g)$, which is a contradiction following Lemma~\ref{fXORg}.
\end{proof}

Next, we explore the tightness of Theorem~\ref{fgunique}.

\begin{proposition}
For every $p,q\in \mathbb{N}$, greater than 1, 
there exists Boolean functions $f$ and $g$ such that $\s(f)=p$, $\s(g)=q$, and $f$ and $g$ coincide on a ball of radius $\s(f)+\s(g)-1$.
\end{proposition}
\begin{proof}
Without loss of generality we will assume that $p\le q$. Fix $p$ and $q$. We will build two function $f$ and $g$ on $p+q$ variables. Let $a\in\{0,1\}^{p+q}$ be a special input defined as follows: $\forall i\in [p+q]$, $a_i=1$ if and only if $i=1$, or $i>2p$, or $i\neq 2p$ is even. Now we define $f$ and $g$ as follows:
$$f(x_1,\ldots ,x_{p+q})=\begin{cases}
0&\text{ if\ \ }\underset{i=1}{\overset{2p}{\mathlarger\sum}}x_i<p\\
1&\text{ if\ \ }\underset{i=1}{\overset{2p}{\mathlarger\sum}}x_i>p\\
\underset{\begin{subarray}{c}
  x_j=1, \\
  j\le 2p
  \end{subarray}}{\mathlarger\sum} j\text{ mod } 2&\text{ if\ \ }\underset{i=1}{\overset{2p}{\mathlarger\sum}}x_i=p
\end{cases}$$

$$g(x)=\begin{cases}
0&\text{ if\ }x=a\\
f(x)&\text{ otherwise.}
\end{cases}$$

Now, we will show that $\s(f)=p$. Fix $x\in\{0,1\}^{p+q}$. $x$ is not sensitive on the last $q-p$ coordinates. If $\underset{i=1}{\overset{2p}{\mathlarger\sum}}x_i<p-1$ or $\underset{i=1}{\overset{2p}{\mathlarger\sum}}x_i>p+1$ then $\s(f,x)=0$. If $\underset{i=1}{\overset{2p}{\mathlarger\sum}}x_i=p$  then $\s(f,x)=p$. If $\underset{i=1}{\overset{2p}{\mathlarger\sum}}x_i=p-1$  then $\s(f,x)\le p$. This is because for any subset of $[2p]$ of size $p+1$ there is both an odd number and an even number in the subset (by pigeonhole principle), and thus amongst its $p+1$ neighbors of hamming weight $p$ (in the first $2p$ coordinates) there must be a neighbor which is not sensitive w.r.t.\ $x$. Similarly, we have that if $\underset{i=1}{\overset{2p}{\mathlarger\sum}}x_i=p+1$  then $\s(f,x)\le p$.  

Next, we will show that $\s(g)=q$. Fix $x\in\{0,1\}^{p+q}$. If $x$ is not in the hamming ball of radius 1 centered at $a$ then, $s(g,x)=s(f,x)\le p \le q$. If $x=a$ then, it is sensitive on all the last $q-p$ coordinates and has $p$ sensitive neighbors in the first $2p$ coordinates. Thus, $\s(g,a)=q$. If $x$ is a neighbor of $a$ through one of the last $q-p$ coordinates (i.e., assuming $q-p>0$) then $\s(g,x)=p+1\le q$. If $x$ is a neighbor of $a$ through one of the first $2p$ coordinates then, we can assume $g(x)=1$. This means hamming weight of $x$ in first $2p$ coordinates is $p+1$ and we know that it is not sensitive on the last $q-p$ coordinates. From the definition of $a$, we know that there is at least one neighbor of $x$ of hamming weight $p$ in the first $2p$ coordinates such that its value on $g$ is the same as $g(x)$. Therefore $\s(g,x)\le p\le q$.   

Finally, we claim that $f$ and $g$ coincide on the ball of radius $p+q-1$ centered at $\left(a\oplus \vec{1}\right)$ (follows from the construction of $g$). 
This completes the proof as $f$ and $g$ are distinct ($f(a)\neq g(a)$) and to distinguish between them by a ball centered at $\left(a\oplus \vec{1}\right)$, we need to consider a ball of radius $p+q$.
\end{proof}

In the case of monotone Boolean functions, we can improve upon the results in Theorem~\ref{funique} and Theorem~\ref{fgunique} as follows: any monotone Boolean function $f$ is uniquely specified by its values on the ball of radius $\s$ centered at $0^n$. This is because, for any input $x$ of hamming weight greater than $\s(f)$, $f(x)$ is equal to 1 if at least one of its neighbors of hamming weight $|x|-1$ is evaluated to 1 on $f$. In other words,
$$f(x)=\underset{\begin{subarray}{c}
  y=x\oplus e_i \\
  |y|=|x|-1
  \end{subarray}}{\mathlarger{\mathlarger{\mathlarger{\mathlarger{\mathlarger\vee}}}}} f(y).$$ 
Furthermore, this result is tight because Wegener's monotone Boolean function \cite{W85} $f$ of sensitivity $\frac{1}{2}\log n+\frac{1}{4}\log\log n +\mathcal{O}(1)$  is identical to the constant zero function on the ball of radius $\s(f) -1$ centered at $0^n$.

\section{The Sensitivity Conjecture: A Computational Perspective}\label{Problem}

We would like to briefly discuss in this section a new perspective on the sensitivity conjecture. Consider a strong version of the sensitivity conjecture which was suggested by Nisan and Szegedy \cite{NS92}: for every Boolean function $f$, we have $\bs(f) \le c\cdot \s(f)^2$, for some constant $c$. Let us assume that the above conjecture is true. We note here that there is no evidence or reason to refute  this  strong version of the sensitivity conjecture. Now consider a computational problem called the \emph{sensitivity problem} defined based on this assumption.

\begin{definition}[Sensitivity Problem]
Given a circuit $C:\{0,1\}^n \to \{0,1\}$, $x \in \{0,1\}^n$, and blocks $B_1, \ldots , B_k$, the sensitivity problem is to find $y \in \{0,1\}^n$ such that $\s(C,y)\ge\sqrt{ \bs(C,x,B_1,...,B_k)/c} $. 
\end{definition}

A solution to the sensitivity problem is guaranteed to exist and a solution can be verified in $\text{poly}(n)$ time, thus the problem is in $\TFNP$. We wonder if the proof of the sensitivity conjecture would give us an efficient algorithm to solve this problem in $\P$? 

We investigated the proofs of the sensitivity conjecture for restricted classes of Boolean functions that exist in literature. In each of these proofs we indeed find an efficient algorithm to solve the above problem in $\P$. For instance, consider the class of Boolean functions admitting the normalized block property. In this case, the computational problem would be that given a DNF $\Phi$, $x \in \{0,1\}^n$, and blocks $B_1, \ldots , B_k$, find   $y \in \{0,1\}^n$ such that $\s(\Phi,y)\ge \sqrt{ \bs(\phi,x,B_1,...,B_k)/2} $ or find two clauses in $\Phi$ which violate $\Phi$ admitting the block property. This problem like the sensitivity problem is in $\TFNP$. However, the proof of Lemma~\ref{zerosbound} gives us an efficient algorithm to find an input $a_1$ with sensitivity $\left\lceil\frac{\dv}{2\da-1}\right\rceil$ and the proof of Lemma~\ref{onesbound} gives us an efficient algorithm to find an input $a_2$ with sensitivity $\left\lceil\frac{\da}{2}\right\rceil$. Since, $\bs(\phi,x,B_1,...,B_k)\le \bs(\Phi)=\dv$, either $a_1$ or $a_2$ is a solution to our problem (assuming there is no violation to the block property of $\Phi$). Thus the computational problem in the case of functions admitting the block property is in $\P$. 

Similarly, for every monotone function $f$, and every input $x$ we have $\s(f,x)=\bs(f,x)$ \cite{N91}. Therefore, for the computational version of the sensitivity problem adapted to the monotone restriction, the input $x$ will be a trivial solution and thus the computational problem would be in $\P$. 
Finally, even for the case of min-term transitive functions, we have an efficient algorithm implicit in the proof of Chakraborthy \cite{S05} who showed that for any min-term transitive function $f$, $\bs(f)\le 2\s(f)^2$.  

Returning to ponder on the existence of efficient algorithms for the sensitivity problem, while it is related to the sensitivity conjecture, it is possible that the sensitivity conjecture is true but there is no efficient algorithm for the sensitivity problem. Similarly, it is possible that an efficient algorithm for the sensitivity problem is found without resolving the sensitivity conjecture (in this case the sensitivity problem should be considered to be in $\NP$ and not in $\TFNP$). However, our progress on the sensitivity conjecture under various restricted settings seem to be by finding a vertex of high sensitivity by starting from a given input with high block sensitivity. Therefore, studying various restrictions on models of computations for Boolean functions seems to be the right direction to pursue, in order to make progress on the sensitivity conjecture.

\section{Conclusion}\label{Conclusion}
In this paper, we motivate the study of the sensitivity conjecture through restrictions on a model of computation. In this regard, we introduced a structural restriction on DNFs representing Boolean functions called the normalized block property. We showed that the examples of Boolean functions that are popular in literature for having a quadratic separation between sensitivity and block sensitivity admit this property. More importantly, we showed that the sensitivity conjecture is true for the class of Boolean functions admitting the normalized block property. Furthermore, we extended a result of Gopalan et al.\ \cite{GNSTW16} and also provided matching lower bounds for our results. Finally, we motivated a new computational problem about finding an input with (relatively) high sensitivity, with respect to the block sensitivity for a given input.

\section{Acknowledgements}
We would like to thank Satya Lokam for discussions and encouraging us to work on the sensitivity conjecture. In particular, Karthik would like to thank him for the internship at Microsoft Research. We would like to thank Avishay Tal for pointing out Corollary~\ref{senslower}. We would like to thank Roei Tell for motivation and discussions. We would like to thank Andris Ambainis and Kri\v{s}j\={a}nis Pr\={u}sis for pointing out an error in a first version of Theorem~\ref{thm_strongAS} and for highlighting the relation between this theorem and the result in~\cite{AP14}. Finally, we would like to thank the anonymous reviewers for helping us improve the presentation of the paper.

\bibliographystyle{alpha}
\bibliography{References}

\appendix
\clearpage

\section{Missing Proofs}

\subsection{The case of Block property, $2$-Mixing property, and Transitivity property}
	\label{2mixing}
We prove here Theorem~\ref{thm_strongAS}.
The proof is an adaptation of the one of Theorem~2 in~\cite{AS11}.
\begin{theorem*}[Restatement of Theorem~\ref{thm_strongAS}]
If $f$ is a Boolean function which depends on at least two variables and represented in a compact form which admits the block property, the $2$-mixing property and the transitive property then, 
\begin{equation} \bs \leq \frac{2}{3}\s^2-\frac{1}{3}\s. \label{eqn_bs-s}
\end{equation}
\end{theorem*}
\begin{proof}
In fact, we will show that 
\begin{claim}\label{clm_caseNoSingleAndGate}
	If all the hypotheses of Theorem~\ref{thm_strongAS} are satisfied and if all the AND gates of the DNF (given by the compact form) depend on at least two variables, we have: 
\begin{align}\label{eqn_s0s1}\bs\le \frac{2}{3}\ones \zeros - \frac{1}{3}\zeros.\end{align}
\end{claim}

First, let us show how the theorem follows from Claim~\ref{clm_caseNoSingleAndGate}. If the DNF (given by the compact form) does not contain an AND gate which depends on exactly one variable, then the theorem is immediate since $\ones,\zeros \le \s$. So, let us assume that there are $p$ AND gates in the DNF which depend on exactly one variable (with $p\geq 1$). Let $\wedge_j$ be such a gate. We know $\wedge_j$ is of the form $x_i$  (as $f(0^n)=0$, the variable has to appear positively). Now, let us assume that the variable $x_i$ appears positively in another AND gate $\wedge_k$ (with $k\neq j$). In this case, the formula does not depend on this new AND gate, more formally, $S_k \subseteq S_j$ which contradicts the fact that $f$ is given in a compact form. Then, let us assume that the variable $x_i$ appears negatively in an AND gate $\wedge_k$. This time it contradicts the $2$-mixing property between the two gates $\wedge_j$ and $\wedge_k$. Consequently, the variable $x_i$ appears only under the gate $\wedge_j$. In particular, after renaming the variables, $f(\mathbf{x})$ is of the form $g(x_1,\ldots,x_{n-p}) \vee x_{n-p+1} \vee \ldots \vee x_n$ where $g$ is a Boolean function represented in the compact form, admitting the block property, the $2$-mixing property and the transitive property and such that any AND gate depends on at least two variables. As we know $\bs(f)=\zerobs(f)$, we have:
\begin{align}
\bs(f) & = \bs_0(g) + \bs_0(x_{n-p+1})+ \ldots +\bs_0(x_n) \nonumber \\
	& \leq  \frac{2}{3} \ones(g)\zeros(g) - \frac{1}{3}\zeros(g) +p \quad (\textrm{By Claim}~\ref{clm_caseNoSingleAndGate}) \nonumber \\
	& = \frac{2}{3} \ones(f) (\zeros(f) -p) - \frac{1}{3}(\zeros(f) -p) +p \nonumber \\
	& = \frac{2}{3} \ones(f)\zeros(f) - \frac{1}{3}\zeros(f)  +p(-\frac{2}{3}\ones(f)+\frac{1}{3}+1). \label{eqn_bs}
\end{align}
Consequently, the theorem holds as soon as $p(-\frac{2}{3}\ones(f)+\frac{1}{3}+1) \leq 0$, i.e., as soon as $\ones(f) \geq 2$. Hence, let us assume that $\ones(f)=1$ (as $f$ is supposed to depend on at least two variables, we know that $\ones(f) \geq 1$). In particular, $\s(f)=\zeros(f)\geq \ones(f)$. Moreover, in the all zeros input, $f$ is sensitive on the last $p$ indices, then $\zeros(f) \geq p$.  By~(\ref{eqn_bs}),
\begin{align*}
\bs(f) & \leq \frac{2}{3}\zeros(f) -\frac{1}{3}\zeros(f) + \frac{2}{3}\zeros(f) \\
	& \leq \frac{4}{3}\zeros(f) - \frac{1}{3}\zeros(f).
\end{align*}
Then~(\ref{eqn_bs-s}) is satisfied as soon as $\zeros(f) \geq 2$. Consequently, the last case corresponds to a Boolean function $f$ such that $\zeros(f)=\ones(f)=1$. This is possible only if $f$ depends on only one variable which is forbidden by hypothesis of the theorem.
\end{proof}

Let us prove now Claim~\ref{clm_caseNoSingleAndGate}, that will finish the proof of the theorem.

\begin{proof}
The transitive property naturally partitions the gates $\wedge_i$'s: $\wedge_i$ and $\wedge_j$ belong to the same component if and only if $(A_i\cup\overline{A}_i)\cap (A_j\cup\overline{A}_j)\neq\emptyset$. Let $k$ be the number of components. We relabel the AND gates using $(i,j)$ to mean that $\wedge_{i,j}$ is the $j^{\text{th}}$ AND gate in the $i^{\text{th}}$ component. More precisely, we have the following gates: $\wedge_{i,1},\ldots ,\wedge_{i,m_i}$, for all $i\in[k]$. Clearly we have that $\sum_{i=1}^km_i=\dv$. 
For every $i\in[k]$ such that $m_i>1$ let $\ell_i$ be the largest integer such that the following holds for all distinct $j_1,j_2\in[m_i]$:
$$\left\lvert\left(A_{i,j_1}\cap \overline{A}_{i,j_2}\right)\cup\left(A_{i,j_2}\cap \overline{A}_{i,j_1}\right)\right\rvert\ge \ell_i.$$
Let us notice that the $2$-mixing property implies that $\ell_i \geq 2$ for all such $i$. Moreover, if $m_i=1$, we will define $\ell_i=3$. Let $I(\ell_i)$ be defined as follows:
$$I(\ell_i)=\begin{cases}
2&\text{ if }\ell_i = 2\\
1&\text{ if }\ell_i\ge 3.\\
\end{cases}$$
We first assume the following relations between the sensitivity and the parameters of our DNF representation:
\begin{claim}\label{zeros=k}
$$\ones\zeros\ge \da\sum_{i=1}^k I(\ell_i) \quad \textrm{and} \quad \zeros \le \left(\sum_{\ell_i=2} m_{i}\right)+\left(\sum_{\ell_i\ge 3} 1 \right).$$
\end{claim}
For all $i_0\in[k]$ such that $m_{i_0}>1$, by the definition of $\ell_{i_0}$, for all distinct $j_1,j_2\in[m_{i_0}]$ we have:
$$\left\lvert\left(A_{i_0,j_1}\cap \overline{A}_{i_0,j_2}\right)\cup\left(A_{i_0,j_2}\cap \overline{A}_{i_0,j_1}\right)\right\rvert\ge \ell_{i_0}.$$
Summing over all distinct $j_1,j_2\in[m_{i_0}]$ we have:
\begin{align}\label{eq_DS}
  \sum_{j_1\neq j_2\in[m_{i_0}]}\left\lvert A_{i_0,j_1}\cap \overline{A}_{i_0,j_2}\right\rvert\ge \ell_{i_0}\cdot \frac{m_{i_0}(m_{i_0}-1)}{2}.
\end{align}
One can notice that (\ref{eq_DS}) is true even if $m_{i_0} =1$. So from now, we can fix any $i_0\in[k]$.

We will show by precise averaging arguments for different cases that:
\begin{claim}\label{cases3}
  There exists $j \in [m_{i_0}]$ such that 
\begin{equation}
\abs{A_{i_0,j}}+\abs{\overline{A}_{i_0,j}} \ge \ell_{i_0}\cdot \frac{m_{i_0}-1}{2}+3-I(\ell_{i_0}). \label{eqn_claim}
\end{equation}
\end{claim}
Then, we have:
$$d_{\wedge_{i_0,j}}\ge \ell_{i_0}\cdot \frac{m_{i_0}-1}{2}+3-I(\ell_{i_0}).$$
From Claim~\ref{zeros=k} we get:
\begin{align*}
\ones\cdot\zeros 
&\geq \sum_{i=1}^k \da\cdot I(\ell_i) \\
&\ge \sum_{i=1}^k \left( \ell_i\cdot \frac{m_{i}-1}{2}+3-I(\ell_{i})\right)\cdot I(\ell_i)\\
&\ge\left(\sum_{\ell_i=2} 2\cdot m_{i}\right)+\left(\sum_{\ell_i\ge 3} \frac{3m_i+1}{2}\right)\\
&= \left(\sum_{i=1}^k \frac{3}{2}\cdot m_i\right)+\left(\sum_{\ell_i=2} \frac{1}{2}\cdot m_{i}\right)+\left(\sum_{\ell_i\ge 3} \frac{1}{2}\right)\\
&= \frac{3}{2}\cdot \dv+ \frac{1}{2}\left(\left(\sum_{\ell_i=2}  m_{i}\right)+\left(\sum_{\ell_i\ge 3} 1\right)\right)\\
&\geq \frac{3}{2}\cdot\bs + \frac{1}{2}\cdot \zeros
\end{align*}
where the last inequality follows from Lemma~\ref{bs=dv} and Claim~\ref{zeros=k}.
\end{proof}

\begin{proof}[Proof of Claim~\ref{zeros=k}]
We note that the transitivity property also $k$-partitions the variable set $X$ into $X_1,\ldots ,X_k$. Therefore, we have that a variable can appear in at most one component. We will build an input $a$ such that $\zeros(a)\ge\sum_{i=1}^k I(\ell_i)$. We describe the construction of $a$ componentwise. Consider the $i^{\text{th}}$ component. We have one of the following cases:
\begin{itemize}
\item[ ]Case 1: $\ell_i\ge 3$. Set all variables \textbf{but one} appearing in $A_{i,1}$ to $1$, and set the rest of variables in $X_i$ to 0.
\item[ ]Case 2: $\ell_i=2$.
This implies that there exists $i\in [k],j_1,j_2\in[m_i]$ such that:
\begin{align*}
  \left(A_{i,j_1}\cap \overline{A}_{i,j_2}\right)\cup\left(A_{i,j_2}\cap \overline{A}_{i,j_1}\right)=\{x_p,x_q\},
\end{align*}
where $x_p,x_q\in X_i$. Set all variables in $X_i\setminus (\{x_p,x_q\}\cup A_{i,j_1}\cup A_{i,j_2})$ to 0. Set all variables in $(A_{i,j_1}\cup A_{i,j_2})\setminus \{x_p,x_q\}$ to 1. Set $x_p$ and $x_q$ such that $\wedge_{i,j_1}(a\oplus e_p)=1$ and $\wedge_{i,j_2}(a\oplus e_q)=1$.
\end{itemize}
Note that $f(a)=0$ as in case 1 none of $\wedge_{i,1}$ evaluate to $1$ and in case 2 both $\wedge_{i,j_1}$ and $\wedge_{i,j_2}$ evaluate to 0 on $a$ by construction, and the remaining AND gates in both cases evaluate to 0 because of the block property and the way we constructed $a$. It is easy to see that in case 1 $a$ has one sensitive bit and in case 2 both $x_p$ and $x_q$ are sensitive. Therefore, we have:
\begin{align*}
  \zeros\ge \zeros(a)=\sum_{i=1}^k I(\ell_i).
\end{align*}

Let $i_0$ such that $d_{\wedge_{i_0}}=d_{\wedge}$. Let $x$ be the input defined by $x_j=1$ if and only if $x_j$ is in $A_{i_0}$. In particular, $\wedge_{i_0}(x)=1$, so $f(x)=1$. Moreover, for all $i \neq i_0$, as $A_i \cap A_{i_0} =\emptyset$ and $A_i \neq \emptyset$ (block property) it implies that $\wedge_i(x)=0$. Let us assume there exists a variable $v$ in $A_{i_0} \cup \bar{A}_{i_0}$ such that $f(x^{(v)})=1$. It implies that there exists $i\neq i_0$ such that $\wedge_i(x^{(v)})=1$ which contradicts the $2$-mixing property. 

Finally, let us assume that the maximal $0$-sensitivity of $f$ is reached on the input $y$. So for any gate $\wedge_i$, there is at most one variable $v$ such that $\wedge_i(y^{(v)}) =1$. Moreover, for every component $i\in[k]$, if $\ell_i \geq 3$, then there is at most one variable from $X_i$ which is sensitive on the input $y$. Consequently, $\zeros(f) \leq \sum_{l_i=2} m_i  +  \sum_{l_i\geq 3} 1$. 
\end{proof}

\begin{proof}[Proof of Claim~\ref{cases3}]
	First, let us notice that if $m_{i_0}=1$, then the left hand side of (\ref{eqn_claim}) is at least equal to two since by hypothesis each AND gate depends on at least two variables and the right hand side evaluates also to two. Hence, we can assume in the following that $m_{i_0}>1$.
	
  Let us consider the property $\mathcal{P}$ which asserts that if we consider Equation~(\ref{eq_DS}) as a sum over $j_2$, all its terms are equal, more formally: for every two $j_1,j_2\in[m_{i_0}]$:
  $$\sum_{j^\prime\in[m_{i_0}]}\left\lvert A_{i_0,j^\prime}\cap \overline{A}_{i_0,j_1}\right\rvert=\sum_{j^\prime\in[m_{i_0}]}\left\lvert A_{i_0,j^\prime}\cap \overline{A}_{i_0,j_2}\right\rvert.$$
  
  We distinguish three different cases:
  \begin{description}
  \item[If $\ell_{i_0}=2$.] Then, by an averaging argument, there exists $j\in[m_{i_0}]$ such that the following holds:
    $$\sum_{j^\prime\in[m_{i_0}]}\left\lvert A_{i_0,j^\prime}\cap \overline{A}_{i_0,j}\right\rvert\ge \ell_{i_0}\cdot \frac{m_{i_0}-1}{2}.$$
    Since $I(\ell_0)=2$, we can rewrite the above as follows:
    $$\sum_{j^\prime\in[m_{i_0}]}\left\lvert A_{i_0,j^\prime}\cap \overline{A}_{i_0,j}\right\rvert\ge \ell_{i_0}\cdot \frac{m_{i_0}-1}{2}+2-I(\ell_{i_0}).$$
    From the block property we have that $\forall j_1,j_2\in[m_{i_0}],$ $A_{i_0,j_1}\cap A_{i_0,j_2}=\emptyset$. This implies: 
    $$\left\lvert \overline{A}_{i_0,j}\right\rvert\ge \ell_{i_0}\cdot \frac{m_{i_0}-1}{2}+2-I(\ell_{i_0}).$$
    From Definition~\ref{cf}a), we have that $A_{i_0,j}$ is non-empty and therefore we have:
    $$\left\lvert \overline{A}_{i_0,j}\right\rvert+\left\lvert A_{i_0,j}\right\rvert\ge \ell_{i_0}\cdot \frac{m_{i_0}-1}{2}+3-I(\ell_{i_0}).$$
    
  \item[If $\ell_{i_0}\ge 3$ and property $\mathcal{P}$ holds.] For every $j \in[m_{i_0}]$ we get:
    $$\sum_{j^\prime\in[m_{i_0}]}\left\lvert A_{i_0,j^\prime}\cap \overline{A}_{i_0,j}\right\rvert \ge \ell_{i_0}\cdot \frac{m_{i_0}-1}{2}.$$
    From the block property we have that $\forall j_1,j_2\in[m_{i_0}],$ $A_{i_0,j_1}\cap A_{i_0,j_2}=\emptyset$. This implies: 
    $$\left\lvert \overline{A}_{i_0,j}\right\rvert\ge \ell_{i_0}\cdot \frac{m_{i_0}-1}{2}.$$
    If we have that $\forall j\in[m_{i_0}], |A_{i_0,j}|=1$ then, it would violate the assumption that $\ell_{i_0}\ge 3$ and thus there must exist $j\in[m_{i_0}]$ such that $|A_{i_0,j}|\ge 2$. Therefore, we have:
    $$\left\lvert \overline{A}_{i_0,j}\right\rvert+\left\lvert A_{i_0,j}\right\rvert\ge \ell_{i_0}\cdot \frac{m_{i_0}-1}{2}+2.$$
    Since $I(\ell_0)=1$, we can rewrite the above as follows:
    $$\left\lvert \overline{A}_{i_0,j}\right\rvert+\left\lvert A_{i_0,j}\right\rvert\ge \ell_{i_0}\cdot \frac{m_{i_0}-1}{2}+3-I(\ell_{i_0}).$$
    
  \item[If $\ell_{i_0}\ge 3$ and property $\mathcal{P}$ does not hold.] It implies there exists $j\in[m_{i_0}]$ such that:
    $$\sum_{j^\prime\in[m_{i_0}]}\left\lvert A_{i_0,j^\prime}\cap \overline{A}_{i_0,j}\right\rvert> \ell_{i_0}\cdot \frac{m_{i_0}-1}{2}.$$
    Since $I(\ell_0)=1$, we can rewrite the above as follows:
    $$\sum_{j^\prime\in[m_{i_0}]}\left\lvert A_{i_0,j^\prime}\cap \overline{A}_{i_0,j}\right\rvert\ge \ell_{i_0}\cdot \frac{m_{i_0}-1}{2}+2-I(\ell_{i_0}).$$
    Still from Definition~\ref{cf}a), we have that $A_{i_0,j}$ is non-empty and therefore we have:
    $$\left\lvert \overline{A}_{i_0,j}\right\rvert+\left\lvert A_{i_0,j}\right\rvert\ge \ell_{i_0}\cdot \frac{m_{i_0}-1}{2}+3-I(\ell_{i_0})$$ which proves the claim.\qedhere
  \end{description}
\end{proof}

\subsection{A framework which contains the result of Ambainis and Sun}
\label{frame}

In~\cite{AS11}, the authors prove that any function $g$ 
such that $\zeros(g)=1$ have the property $\ones(g) \geq 3\frac{\zerobs(g)-1}{2}$.  

We show here that some small variations of these functions (with the same sensitivity and block sensitivity) admit the block property, the transitive property, and the $3$-mixing property and so a weakly stronger lower bound is directly implied by Theorem~\ref{thm_strongAS}.

\begin{definition}
  We consider the natural representation of the set $\{0,1\}^n$ as a graph where the vertices are the points of $\{0,1\}^n$ and there is an edge between two vertices $x$ and $y$ if and only if $x$ and $y$ are at Hamming distance one. We call this graph the $n$-dimensional Boolean cube.
  Let $g:\{0,1\}^n \rightarrow \{0,1\}$ be a Boolean function. The $1$-set of $g$ is the induced subgraph $G$ of the $n$-dimensionnal Boolean cube where the vertices of $G$ are exactly the vectors $x\in \{0,1\}^n$ such that $f(x)=1$.
\end{definition}

The following claim is implicit in the proof of Theorem~2 in~\cite{AS11}.
\begin{claim}\label{hypercube}
  If $g:\{0,1\}^n \rightarrow \{0,1\}$ is a non-constant Boolean function such that $\zeros(g)=1$, then the connected components of the $1$-set of $g$ are hypercubes and the distance (in the Boolean cube) between two such hypercubes is at least three.
\end{claim}

\begin{proof}
  Let $g$ be such a Boolean function. If for some $x\in \{0,1\}^n, i \neq j\in [n]$ we have that $g(x)=g(x^{i})=g(x^{j})=1$, then $g(x^{\{i,j\}})=1$ (otherwise $\zeros(g,x^{\{i,j\}}) \geq 2$), which implies that each component of the $1$-set is a hypercube. If two hypercubes are at distance one, it means they are in the same component which contradicts the previous statement. Finally, if two hypercubes are at distance two, it implies there exists $x$ in the neighborhood of both  hypercubes. However, in this case the zero-sensitivity in $x$ is at least two which contradicts the hypothesis. Consequently, the $1$-set of the function $g$ is a union of disjoint hypercubes such that any two of them are at distance at least three.
\end{proof}

\begin{claim*}[Restatement of Claim~\ref{ASinFramework}]
  If $f$ and $g$ are defined as in Theorem~\ref{AS}, then there exists $f^\prime$ with same block sensitivity and at most same $1$-sensitivity than $f$ which admits the block property, the transitive property, and the $3$-mixing property.
\end{claim*}

\begin{proof}
  Let $f$ and $g$ be such Boolean functions. The function $f$ is a disjoint-variables union of functions of $g$. We can first notice that if $g$ has the three properties, than it is the same for $f$. 
  
  Let us show now that $g$ admits the three properties. By negation of some variables, we assume that the maximal $0$-block sensitivity of $g$ is reached on the all zeros input. Let us denote by $B_1,\ldots,B_{\bs}$ the minimal associated disjoint blocks. By Claim~\ref{hypercube}, we know that the $1$-set of $g$ is a union of disjoint hypercube such that two of them are at distance at least three (this remark already appeared in the proof of Theorem~2 in~\cite{AS11}). In particular 
  \begin{align*}
    g^{-1}(1) = \bigcup_{i=1}^p S_i 
  \end{align*}
  where each $S_i$ is a hypercube of the form
  \begin{align*}
    S_i = \{ (x_1,\ldots,x_n) \in \{0,1\}^n \mid x_{ij_1}=\ldots=x_{ij_{l_i}}=0, x_{ik_1}=\ldots=x_{ik_{m_i}}=1 \}
  \end{align*} and the $S_i$'s are at distance at least three.
  Let us notice, that for any block $B_i$ there is an hypercube $S_{\pi(i)}$ such that $B_i={x_{\pi(i)k_1},\ldots,x_{\pi(i)k_{m_{\pi(i)}}}}$. We can relabel the blocks such that for all $i\leq \bs$, $\pi(i)=i$. Then we can remove all $S_j$ with $j> \bs$ from the $1$-set of $g$, we get a new function (block sensitivity does not change and the $1$-sensitivity can only decrease)
  \begin{align*}
    g^{-1}(1) = \bigcup_{i=1}^\bs S_i. 
  \end{align*}
  Let us show that the DNF 
  \begin{align*}
    \bigvee_{i=1}^\bs ( \bar{x}_{ij_1} \wedge \ldots \wedge \bar{x}_{ij_{l_i}} \wedge x_{ik_1} \wedge \ldots \wedge x_{ik_{m_i}})
  \end{align*} satisfies the three properties. As all the positive parts $\{ x_{ik_1},\ldots, x_{ik_{m_i}} \}$ correspond to the disjoint blocks $B_i$ and as $g(0)=0$, this representation has the block property. As the hypercubes are disjoint, they share pairwise at least one variable which ensures the transitive property. Finally, as seen before, the hypercubes are at distance at least three, and so the representation has the $3$-mixing property.
\end{proof}

\subsection{Examples in Literature having Block property}\label{sec_gap}

Several examples of Boolean functions in the literature achieve a quadratic gap between sensitivity and block sensitivity. We proved in Section~\ref{sec_BP} that the one introduced by Ambainis and Sun falls within our framework (the function has block property, transitive property and $3$-mixing property). We show here that this fact is also true for other examples of such functions. 

Before Ambainis and Sun's result, Rubinstein~\cite{R95} exhibited a Boolean function $f$ such that $\bs=\frac{1}{2}\s^2$. As usual, the function is a variables-disjoint union of a function $g$: 
\begin{align*}
  f(x) = \bigvee_{i=1}^{2n} g(x_{i,1},\ldots,x_{i,2n}).
\end{align*} 
The $(2n)$-variate function $g$ evaluates to one on the input $x$ if there exists an even index $(2j)$ such that $x_{2j}=x_{2j-1}=1$ and all other variables are set to zero. In particular all the $\underline{11}$-blocks for different patterns are disjoint, so the function admits the block property. Given $i$, all the patterns on the $(x_{i,j})_{j\leq 2n}$ variables pairwisely intersect, which ensures the transitive property. Finally, two patterns disagree on exactly four variables, so the function has the $4$-mixing property. In fact, to get the best bound, we showed that we only need the $3$-mixing property. Ambainis and Sun's example directly comes from Rubinstein's one by removing some constraints in the patterns such that two patterns disagree only on three variables.  

Another example, with $\bs=\frac{1}{2}\s^2+\frac{1}{2}\s$, was introduced by Virza~\cite{V11}. The idea is to add, in Rubinstein's example, one pattern with only a block $\underline{1}$ as positive part of the pattern. More formally, $g$ depends on $2n+1$ variables and evaluates to one on $x$ either if there exists an even index $(2j)$ such that $x_{2j}=x_{2j-1}=1$ and all other variables are set to zero (same patterns than in Rubinstein's example) or if $x_{2n+1}=1$ and all other variables are set to zero. By similar reasons, this function has the block and the transitive properties. If we consider two patterns of the first form (the ones already present in Rubinstein's), they still disagree on at least four variables. But if we consider one pattern of the first form and the new pattern, they disagree only on three variables. Particularly, this example has the $3$-mixing property and is close to have the $4$-mixing property, and by this way, fits between the examples given by Rubinstein and by Ambainis and Sun. Finally, we note that Chakraborty's example \cite{S05} is obtained by a simple modification of Rubinstein's example to ensure that the function is cyclically invariant, and the arguments made for the example of Rubinstein holds in this case.
\end{document}